\documentclass[fleqn,envcountsame]{LMCS}

\def\doi{8(3:24)2012}
\lmcsheading%
{\doi}
{1--15}
{}
{}
{May~\phantom{.0}2, 2011}
{Sep.~27, 2012}
{}

\usepackage{verbatim} % \begin-\end{comment} to comment out several lines
\usepackage[latin1]{inputenc}% to use the german umlaute
\usepackage[usenames]{color}% for defining our own colors

\usepackage{amsmath}% needed for \binom
\usepackage{amssymb}% needed when using symbols like \mathbb{N}
\usepackage{enumerate,hyperref}

\usepackage{eepic}
\usepackage{xspace}

\definecolor{darkred}{rgb}{0.65,0,0}
\definecolor{darkblue}{rgb}{0,0,0.4}
\definecolor{owngreen}{rgb}{0,0.6,0}
\definecolor{linkdarkgreen}{rgb}{0,0.45,0}

\newcommand{\dotcup}{\mathrel{\mathaccent\cdot\cup}}

\newcommand{\aut}[1]{\ensuremath{\mathcal{#1}}}

\newcommand{\G}{\ensuremath{\Gamma}\xspace}

\newcommand{\pI}{Player~I\xspace}
\newcommand{\pO}{Player~O\xspace}

\newcommand{\Leq}{\ensuremath{L_{\mathrm{eq}}}\xspace}
\newcommand{\GSG}{\ensuremath{\Gamma^{SG}}\xspace}
\newcommand{\Gd}[1]{\ensuremath{\Gamma_{#1}}\xspace}
\newcommand{\GdL}[1]{\ensuremath{\Gamma_{#1}(L)}\xspace}
\newcommand{\GdLA}[1]{\ensuremath{\Gamma_{#1}(\LA)}\xspace}
\newcommand{\Gp}[1]{\ensuremath{\Gamma'_{#1}}\xspace}
\newcommand{\GpLA}[1]{\ensuremath{\Gamma'_{#1}(\LA)}\xspace}
\newcommand{\Gpempty}{\ensuremath{\Gamma'}\xspace}

\newcommand{\LA}{\ensuremath{L(\aut{A})}\xspace}
\newcommand{\LsA}{\ensuremath{L_*(\aut{A})}\xspace}
\newcommand{\LoA}{\ensuremath{L_{\omega}(\aut{A})}\xspace}

\newcommand{\const}[1]{\ensuremath{\langle#1\rangle}\xspace}
\newcommand{\Nat}{\ensuremath{\mathbb{N}}\xspace}
\newcommand{\Natp}{\ensuremath{\mathbb{N}_+}\xspace}
\newcommand{\B}{\ensuremath{\mathbb{B}}\xspace}
\newcommand{\Bst}{\ensuremath{\mathbb{B}^*}\xspace}
\newcommand{\Bsq}{\ensuremath{\mathbb{B}^2}\xspace}
\newcommand{\Bom}{\ensuremath{\mathbb{B}^{\omega}}\xspace}
\newcommand{\Bsqst}{\ensuremath{(\mathbb{B}^2)^*}\xspace}
\newcommand{\Bsqom}{\ensuremath{(\mathbb{B}^2)^{\omega}}\xspace}
\renewcommand{\S}{\ensuremath{\Sigma}\xspace}
\newcommand{\Sst}{\ensuremath{\Sigma^*}\xspace}
\newcommand{\Som}{\ensuremath{\Sigma^{\omega}}\xspace}
\newcommand{\SI}{\ensuremath{\Sigma_I}\xspace}
\newcommand{\SO}{\ensuremath{\Sigma_O}\xspace}

\newcommand{\al}{\ensuremath{\alpha}\xspace}
\newcommand{\be}{\ensuremath{\beta}\xspace}

\newcommand{\ie}{i.e.\xspace}
\newcommand{\cf}{cf.\xspace}
\newcommand{\eg}{e.g.\xspace}

\newcommand{\wlofg}{w.l.o.g.\xspace}

\newcommand{\Inf}{\mathrm{Inf}}
\newcommand{\Index}{\mathrm{index}}
\newcommand{\CFLom}{\ensuremath{\mathrm{CFL}_{\omega}}}

\let\obinom\binom
\renewcommand\binom[2]{
  \Big( { {{#1}} \atop {{#2}} } \Big)
}

\begin{document}
\title[Degrees of Lookahead in Regular~Infinite~Games]{Degrees of Lookahead in Regular~Infinite~Games}

\author[M.~Holtmann]{Michael Holtmann\rsuper a}
\address{{\lsuper a}Lehrstuhl f{\"u}r Informatik 7, RWTH Aachen University}
\email{holtmann@automata.rwth-aachen.de}

\author[{\L}.~Kaiser]{{\L}ukasz Kaiser\rsuper b}
\address{{\lsuper b}LIAFA, CNRS \& Universit{\'e} Paris Diderot -- Paris 7}
\email{kaiser@liafa.univ-paris-diderot.fr}

\author[W.~Thomas]{Wolfgang Thomas\rsuper c}
\address{{\lsuper c}Lehrstuhl f{\"u}r Informatik 7, RWTH Aachen University}
\email{thomas@automata.rwth-aachen.de}


\begin{abstract}
We study variants of regular infinite games where the strict alternation
of moves between the two players is subject to modifications.
The second player may postpone a move for a finite number of steps,
or, in other words, exploit in his strategy some lookahead on the moves
of the opponent. This captures situations in distributed systems, \eg
when buffers are present in communication or when signal transmission
between components is deferred. We distinguish strategies with
different degrees of lookahead, among them being the continuous and
the bounded lookahead strategies. In the first case the lookahead is of finite
possibly unbounded size, whereas in the second case it is of bounded size.
We show that for regular infinite games the solvability by continuous strategies
is decidable, and that a continuous strategy can always be reduced to one
of bounded lookahead. Moreover, this lookahead is at most doubly exponential
in the size of a given parity automaton recognizing the winning condition.
We also show that the result fails for non-regular games
where the winning condition is given by a context-free $\omega$-language.
\end{abstract}

\maketitle

%--------------------------------------------Section 1
\section{Introduction}\label{sec:introduction}
The algorithmic theory of infinite games is a powerful and flexible framework
for the design of reactive systems (see e.g. \cite{GTW02AutLogInfGam}).
It is well known that, for instance, the
construction of a controller acting indefinitely within its
environment amounts to the computation of a winning strategy in an
infinite game. For the case of regular games, algorithmic solutions
of this synthesis problem have been developed, providing methods for
automatic construction of controllers. The basis of this approach is
the B\"uchi-Landweber Theorem, which says that in a regular infinite
game, \ie a game over a finite arena with a winning condition given by
an $\omega$-regular language, a finite-state winning strategy for the
winner can be constructed \cite{BL69SolSeqCondFinStateStr}.
Much work in the past two decades has been devoted to generalizations of
this fundamental result. The game-theoretic setting is built on
two components, a \emph{game arena} or game graph, representing the
transition structure of a system, and a \emph{winning condition},
usually given by a logic formula or an automata theoretic condition.
Most generalizations address an extension of either of the two, or both.
A rapidly growing literature is thus concerned with the case of infinite game
graphs and non-regular winning conditions \cite{Wal96PushProc,Cach03HighOrdPushAutCaucalHier,BSW03PushGamUnboundRegCond}.

In the present paper we investigate a different kind of generalization of
the basic setting, regarding the possibility to get a lookahead on the moves
of the opponent. To explain this aspect it is convenient to refer to
the simplest format of infinite games, also called Gale-Stewart
games \cite{Mosch80DST}.
In such a game we abstract from arenas but just let the two players choose
letters from a finite alphabet in turn. (For notational convenience let us
only consider the typical case of the Boolean alphabet $\B:=\{0,1\}$.)
A play is built up as a sequence $a_0b_0a_1b_1\cdots$ where
$a_i$ is chosen by one player and $b_i$ by the other.
A natural view is to consider the sequence
$\al=a_0a_1\cdots$ as \emph{input stream} and $\beta=b_0b_1\cdots$ as
\emph{output stream}. Accordingly, the players are called
Player Input and Player Output, or short \pI and \pO.
The play is won by \pO if the $\omega$-word
$\obinom{a_0}{b_0}\obinom{a_1}{b_1}\obinom{a_2}{b_2}\cdots\in\Bsqom$
satisfies the winning condition, \ie if it belongs to a given $\omega$-regular
language $L$. In the classical setting, a strategy for \pO is
a function $f$ that maps a finite input prefix $a_0\cdots a_i$ to
the bit $b_i$ that is to be chosen by \pO. Such a strategy induces an operator
$\lambda:\Bom\to\Bom$ from input streams to output streams.~In this
work we study more generalized operators that correspond to strategies
where the choice of $b_i$ depends on $a_0\cdots a_j$, for $j\neq i$.
We show results on the existence of such strategies
for different conditions on the relation between $i$ and $j$.

There are two motivations for the study of such a generalization,
a practical and a theoretical one. In many scenarios, the occurrence of
delays (say between input and output) is realistic, either as a
modeling assumption or as a feature of strategies. For example,
the design of a controller may involve a buffer that allows to store
a sequence of input bits of some fixed length $d$ such that the bit
$b_i$ of the output sequence is to be delivered with lookahead $d$,
\ie on the basis of the input sequence $a_0\cdots a_{i+d}$. Conversely,
in the context of networked control (\ie systems with components in
different locations), there may be a delay $d$ in the transmission of
data, which means that the delivery of $b_i$ is due at a point where
only the input bits $a_0\cdots a_{i-d}$ are available. It is clear
that the occurrence of lookaheads and delays influences the existence
of solutions. In the first case, we obtain for increasing $d$ an increasing
advantage for the output player, whereas in the second case we obtain
an increasing disadvantage. Observe that the cases are symmetric
in the two players, and thus are mutually reducible.

A more theoretical motivation is to explore more comprehensively
and systematically the solution concepts for infinite games.
The classical concept of a strategy gives a very special kind of operator,
but there are natural options of higher generality, well-known
already from background fields like descriptive set theory and topology \cite{Mosch80DST}.
Let us mention four fundamental levels of operators, corresponding to
different levels of obligation for \pO to move. The most general ones are
the continuous operators (see \eg \cite{TB73FinAutBehSynth,TL93LogSpecInfComp}).
An operator $\lambda$ is continuous (in the Cantor space of infinite
sequences over \B) if in the output sequence $\beta = \lambda(\al)$
the bit $b_i$ is determined by a finite prefix of $\al$.
Referring only to the length of prefixes, we call an operator
uniformly continuous if for some strictly monotone function $h:\Nat\to\Nat$ we have that $b_i$
is determined by $a_0\cdots a_{h(i)}$. For fixed $h$ we then speak of $h$-delay operators.
On a further level of specialization, we are dealing with operators of
bounded delay. These are $h$-delay operators with $h(i)\leq i+d$, for some $d\in\Nat$.
Analogously, if $h(i)=i+d$, then we speak of operators with constant delay $d$, and
finally, the function $h(i)=i$ supplies the operators induced by standard strategies.
All these levels of delay naturally correspond to different types of games;
for example, a continuous strategy involves the moves ``wait'' or
``output~$b$'' after each move of the opponent.

Our main result connects the different kinds of operators
in the context of infinite games. We show that
in a two-person game with regular winning condition, one can decide whether
there is a continuous winning strategy for \pO, and in this case 
a strategy of constant delay can be constructed.
Moreover, one can compute a suitable bound $d$ for the delay.
Thus, in the first mentioned application scenario, if a standard
controller for satisfying a regular specification does not exist then
one can decide whether some finite buffer will help, and determine
the needed size of that buffer. We also show that the result fails when
passing to non-regular specifications. However, which functions may be
appropriate for uniformly continuous strategies in the non-regular
case is left open. It seems that for infinite-state (or non-regular)
games our result can serve as an entry into a much wider field of study
(see \eg the recent work \cite{FLZ11}).

As indicated above, the idea of generalized concepts of strategies is
far from new. An early contribution is found in the (not well-known)
paper of Hosch and Landweber \cite{HL72FinDelSol}. It deals with constant
delay strategies in regular games and exploits a result of Even and Meyer
from Boolean circuit theory
to establish a bound for delays \cite{EM69SeqBoolEq}.
We obtain this result here as a corollary of the main theorem.
The extension of our result over \cite{HL72FinDelSol} covers
three aspects: the connection with strategies of unbounded delay,
a considerably simplified and transparent proof of the Hosch-Landweber-Theorem
(the construction in \cite{HL72FinDelSol} is highly complex),
and finally better complexity bounds for suitable delays.

This paper is organized as follows.
In the next section we introduce notation.
In Section~\ref{sec:operators_games_delay} we present several kinds
of functions and the operators they induce.
We also bridge from continuous operators to delay operators
and introduce games with delay.
In Sections~\ref{sec:block_game}--\ref{sec:connection_block_semigroup_game}
we prove our main result via a two-stage reduction:
In Section~\ref{sec:block_game} we do the first step,
switching over to block games.
In Section~\ref{sec:semigroup_game} we deal with notions
related to semigroups
and define a semigroup game.
This framework is finally used in Section~\ref{sec:connection_block_semigroup_game}
to establish the second step of the reduction, \ie the connection between
block games and the semigroup game.
Sections~\ref{sec:pusdown_games_delay} and \ref{sec:conclusion}
provide evidence that our results cannot be generalized
to $\omega$-context-free specifications and give an outlook
on future investigations.

%--------------------------------------------Section 2
\section{Preliminaries}\label{sec:preliminaries}
Let \S be a finite \emph{alphabet}.
By \Sst and \Som we denote the sets of finite and infinite \emph{words} over \S.
Usually, finite words are denoted $u,v,\ldots$ whereas $\al,\beta,\ldots$ are infinite words.
By $|u|$ we denote the \emph{length} of $u$ and
$\S^n:=\{u\mid|u|=n\}$ is the set of words of length $n$.
$\Nat$ is the set of natural numbers and $\Natp:=\Nat\setminus\{0\}$.
Given $n_1,n_2\in\Nat$ with $n_1<n_2$
we write $\S^{[n_1,n_2]}$ for $\bigcup_{n_1\leq n\leq n_2}\S^n$.

A \emph{(deterministic) finite automaton}, DFA for short, over \S is a tuple
$\aut{A}=(Q,q_0,\delta,F)$ where $Q$ is a (non-empty) finite set of
\emph{states}, $q_0\in Q$ is the \emph{initial state}, $\delta:Q\times\S\to Q$
is a \emph{transition function}, and $F\subseteq Q$ is a set of
\emph{final states}. The \emph{run} $\rho_u$ of \aut{A} on $u:=u_0\cdots u_{n-1}$
is the finite sequence $\rho_u(0)\cdots\rho_u(n)$ with $\rho_u(0)=q_0$ and
$\rho_u(i+1)=\delta(\rho_u(i),u_i)$ for $i=0,\ldots,n-1$. We define
\aut{A} to \emph{accept} $u$ if and only if $\rho_u(n)\in F$. The
set of all words accepted by \aut{A} is called the
\emph{$*$-language} of \aut{A} and denoted $L_*(\aut{A})$. Later in
our work we need the following basic property of deterministic
finite automata.

\begin{lem}\label{lem:length}
Let \aut{A} be a DFA with $n$ states and $|\LsA|=\infty$. Then, for all
$i\in\Nat$, \aut{A} accepts a word $u_i$ of length $i\leq|u_i|\leq i+n$.
\end{lem}

\begin{proof}
Let \aut{A} be a DFA with $n$ states and $|\LsA|=\infty$. Since \LsA
is infinite it must be possible, for each $i\in\Nat$, to read a word
$u$ of length $i$ such that from $\delta^*(q_0,u)$ a final state is
reachable. Otherwise, the length of words accepted by \aut{A} is
bounded by $i$, which is a contradiction to the infiniteness of
\LsA. Then, from $\delta^*(q_0,u)$ we can reach a final state by a
word $u'$ of length at most $n$. The word $uu'$ is accepted by \aut{A}
and is of length between $i$ and $i+n$.
\end{proof}

A \emph{(deterministic) parity automaton}, DPA for short,
over \S is similar to a DFA, but instead of the set $F$ of final
states it has a \emph{coloring}, \ie a function $c : Q \to \{0,\ldots,m\}$.
A run of a DPA is the natural extension of a run of a DFA to infinite words.
For $\al\in\Som$, the set $\Inf(\rho_{\al})$ is
the set of states visited infinitely often in run $\rho_{\al}$.
We define the parity automaton \aut{A} to accept $\al$ if and only if
$\max(\Inf(c(\rho_{\al})))$ is even, \ie the maximal color seen
infinitely often in the run on $\al$ is even. Accordingly, the
acceptance condition of \aut{A} is called a \emph{max-parity}
acceptance condition. The set of all words accepted by \aut{A} is
called the $\omega$-language of \aut{A} and denoted \LoA.

In the next sections, we write \LA instead of \LsA or \LoA if it is
clear from the context whether \aut{A} is a DFA or DPA. It is
well-known that languages accepted by DPAs are exactly the
\emph{$\omega$-regular} languages (see \eg \cite{GTW02AutLogInfGam}).

A \emph{parity game} $\G = (V,V_I,V_O,E,c)$ is played by two players,
\pI and \pO, on a directed graph $G = (V, E)$:
\begin{iteMize}{$\bullet$}
\item $V = V_I\dotcup V_O$ is a partition of $V$ into positions of \pI and \pO,
\item $E \subseteq V\times V$ is the set of allowed \emph{moves}, and
\item $c:V\to\{0,\ldots,m\}$ is a coloring of $V$
  (w.l.o.g.\ $m\in2\Nat$).
\end{iteMize}
We assume that for each $v \in V$ there is a valid move from $v$,
\ie $vE:=\{w\mid(v,w)\in E\}\neq\emptyset$.
A \emph{play} is an infinite path through $G$. A (standard) \emph{strategy} for \pO is
a function $f:V^*V_O\to V$ defining, for each position of \pO and each history
$v_0\cdots v_k$ of the play, her next move. Thus, for each $v_0\cdots v_k$
(with $(v_i,v_{i+1})\in E$ for all $i=0,\ldots,k-1$) and $v_k\in V_O$, the
function $f$ is defined such that $(v_k,f(v_0\cdots v_k))\in E$.
A play $v_0v_1\cdots$ is \emph{consistent} with the strategy $f$ if for each
$v_i\in V_O$ the next position is given by $f$, \ie $v_{i+1}=f(v_0\cdots v_i)$.

The \emph{parity winning condition} is again defined so that a play
$v_0v_1\cdots$ is winning for \pO if and only if the maximal color occurring
infinitely often in $\{c(v_i) \mid i \in \Nat\}$ is even.
In the other case the play is winning for \pI.
The function $f$ is called a \emph{winning strategy for \pO from $v_0$}
if each play starting in $v_0$ that is consistent with $f$ is winning for \pO,
and analogously for \pI. Parity games, even on infinite graphs, are
\emph{determined}, \ie for each $v$ either \pI or \pO has a winning strategy
from $v$ (see \eg \cite{GTW02AutLogInfGam}).

For the rest of this paper, let us fix $\{0,1\}$ as input and
output alphabet, \ie let $\SI=\SO:=\B$. All the definitions
and results are analogous for other finite alphabets of size
at least two.

%--------------------------------------------Section 3
\section{Operators and Games with Delay}\label{sec:operators_games_delay}

In this section we introduce different kinds of functions and operators,
and show how they induce games with different degrees of lookahead.
Below, we mostly use the term ``delay'' in place of ``lookahead'',
following \eg \cite{HL72FinDelSol}.

%--------------------------------------------Section 3.1
\subsection{Delay Operators}\label{subsec:delay_operators}

Let $\lambda$ denote a function from \Bom to \Bom, also called
an \emph{operator}. We shall distinguish the following classes of operators,
starting form the most general ones.
\begin{enumerate}[(1)]
\item \emph{continuous operators}
\item \emph{uniformly continuous operators}
\item \emph{$h$-delay operators} for a fixed $h:\Nat\to\Nat$
\item \emph{bounded delay operators}
\item \emph{$d$-delay operators} for a fixed $d\in\Nat$
\end{enumerate}

An operator $\lambda$ is continuous if in the output sequence
$\beta = \lambda(\al)$ each bit is determined by a finite prefix of $\al$. 
This condition is equivalent to the standard topological definition,
where $\lambda$ is continuous if the preimage $\lambda^{-1}(U)$ of every open set
$U\subseteq\Bom$ is open in \Bom. Here, open sets in \Bom are given by
the standard Cantor topology, \ie $U\subseteq\Bom$ is open if there exists
$W \subseteq \Bst$ such that $U = \{w\Bom\mid w\in W\}$.
Consult \eg \cite{TL93LogSpecInfComp} for more details.
This topology is induced by the standard metric $\delta$ on \Bom:
\[ \delta(\al,\be) = \begin{cases}
  2^{-\min\{n\mid\al_n\neq\be_n\}} & \text{if }\al\neq\be,\\
  0 & \text{otherwise},
\end{cases} \]
and the standard metric definitions of continuity and uniform continuity
are equivalent to the ones we use. Let us recall here three of these classical
definitions. An operator $\lambda \ :\ \Bom \to \Bom$ is: 
\begin{iteMize}{$\bullet$}
\item \emph{Continuous} if for all $\al,\be\in\Bom$ and each $\epsilon > 0$
  there exists a $\delta > 0$ such that if $\delta(\al,\be) < \delta$ then 
  $\delta(\lambda(\al),\lambda(\be)) < \epsilon$.
\item \emph{Uniformly continuous} if for each $\epsilon > 0$
  there exists a $\delta > 0$ such that for all $\al,\be\in\Bom$
  if $\delta(\al,\be) < \delta$ then 
  $\delta(\lambda(\al),\lambda(\be)) < \epsilon$.
\item \emph{Lipschitz continuous} with constant $C$ if
  for all $\al,\be\in\Bom$ the following holds:
  $\delta(\lambda(\al),\lambda(\be)) \leq C\cdot\delta(\al,\be)$.
\end{iteMize}

\noindent Since we do not use metric properties of the Cantor space,
to formally capture the constraint that each output bit is
determined by a finite prefix of the input, we define
the continuity of $\lambda$ in the following equivalent way.
We use a map $l$ that transforms each input bit into either $0$ or $1$
or $\rhd$, the latter meaning that the production of the next output bit
is still deferred. The value $\lambda(\al)$ is then obtained from the sequence
of $l$-values by deleting all entries $\rhd$.

\begin{defi}\label{def:continuous_to_delay}
An operator $\lambda:\Bom\to\Bom$ is \emph{continuous}
if there exists $l:\Bst\to\{0,1,\rhd\}$
such that for all $\al\in\Bom$ the word
$l(\al):=l(\al_0)l(\al_0\al_1)l(\al_0\al_1\al_2)\cdots$
satisfies the following:
\begin{enumerate}[(1)]
\item $l(\al)$ does not end with $\rhd^\omega$, and
\item $\lambda(\al) = \text{strip}(l(\al))$ where $\text{strip}(l(\al))$ is
the word $l(\al)$ with all $\rhd$ removed.
\end{enumerate}
\end{defi}

Let us now define \emph{$h$-delay} and \emph{uniformly continuous}
operators. Let $h : \Nat\to\Nat$ be a strictly monotone function.
We say that $\lambda$ is an \emph{$h$-delay} operator if,
for each $\al\in\Bom$, the bit $(\lambda(\al))_i$ depends only
on $\al_0\cdots\al_{h(i)}$. An operator $\lambda$ is uniformly continuous
if there exists an $h$ such that $\lambda$ is an $h$-delay operator.
Observe that each uniformly continuous operator is indeed continuous --
the function $h$ supplies the information 
how long the output $\rhd$ should be produced. 
%To see this, simply label the first $f(0)-1$ levels of the binary tree with
%$\perp$ and then each node $\al_0\cdots\al_{f(0)}$ on level $f(0)$
%with the appropriate first bit of $\lambda(\al_0\cdots \al_{f(0)}\beta)$.
%Note that, for any $\beta$, this bit only depends on $\al_0\cdots \al_{f(0)}$.
%Then label levels $f(0)+1$ to $f(1)-1$ with $\perp$
%and level $f(1)$ with the second bits, and so on.

For the space \Bom it is known that
the converse also holds. This is a consequence of K\"onig's Lemma, or
equivalently of the fact that continuous functions
on a closed bounded space are uniformly continuous.

\begin{lem}\label{lem:continuous_to_delay_operator}
For every continuous operator $\lambda:\Bom\to\Bom$ there exists
a strictly monotone function $h:\Nat\to\Nat$ such that $\lambda$ is
an $h$-delay operator.
\end{lem}

By the above lemma, the classes of continuous operators $\Bom\to\Bom$ and
uniformly continuous operators $\Bom\to\Bom$ are exactly the same.
A space where this does not hold is \eg $\mathcal{R}:=\Bom\setminus\{0^{\omega}\}$.
Consider $\lambda_1:\mathcal{R}\to\mathcal{R}$ with
\[ \lambda_1(\al):=\begin{cases}
  01^{\omega} & \text{if }\al=0^*10\be\text{ for some }\be\in\Bom\\
  1^{\omega} & \text{otherwise}
\end{cases} \]
Intuitively, the operator $\lambda_1$ checks if there is $0$ or $1$
after the first $1$ in the input. One can verify that $\lambda_1$ is
a continuous function from $\Bom\setminus\{0^{\omega}\}$ to \Bom, but
it is not uniformly continuous and can not be extended to any continuous
function from \Bom to \Bom. Our results do not hold for such operators:
Already $\lambda_1$ is a counterexample, since it is continuous
but not of bounded delay. Thus, in this paper we adhere to the space \Bom.

Among the uniformly continuous operators, we distinguish an even more
restricted class of bounded delay operators. A function $h:\Nat\to\Nat$
is said to be of \emph{bounded delay} if there exist $i_0,d\in\Nat$
such that $h(i)=i+d$ for all $i\geq i_0$, and it is said to be
a \emph{$d$-delay} function (or a function of \emph{constant} delay $d$)
if $h(i)=i+d$ for all $i\in\Nat$. The induced operators are named accordingly.

In topological terms, bounded delay operators are Lipschitz continuous
functions from \Bom to \Bom, as defined above. The $d$-delay operator is
clearly Lipschitz continuous with constant $C = 2^{d}$.
Conversely, if an operator $\lambda$ is not of bounded delay
then for each $d$ there exists $\al\in\Bom$ and an index $i$
such that the $i$-th bit of $\lambda(\al)$ is not a function of the first
$i+d$ bits of \al. This means that there exists $\be\in\Bom$ with the same
first $i+d$ bits as \al, \ie satisfying $\delta(\al,\be)<2^{-(i+d)}$,
such that $\lambda(\be)$ differs from $\lambda(\al)$ on the $i$-th bit,
therefore $\delta(\lambda(\al),\lambda(\be))\geq2^{-i}$. This contradicts 
Lipschitz continuity as the constant $C$ would have to satisfy $C>2^{d}$,
for all $d\in\Nat$.

In all definitions above, we assume that the delay function
$h$ is strictly monotone. For our purpose it is more convenient
to consider the function $f_h:\Nat\to\Natp$,
denoting the number of additional input bits until the next output bit:
\[ f_h(i) := \begin{cases}
    h(0)+1 & \text{if }i=0\\
    h(i)-h(i-1) & \text{if }i>0
\end{cases} \]
In the next sections, we work only with the functions $f_h$.
Moreover, we use the special notation \const{d} for
the function $f_h$ with $h$ of constant delay $d$:
$\const{d}(0)=d+1$ and $\const{d}(i)=1$ for $i>0$.
From now on, we omit the subscript $h$ in our notation.

%--------------------------------------------Section 3.2
\subsection{Regular Games with Delay}\label{subsec:regular_games_delay}

In this section we introduce the regular infinite game \GdL{f}.
It is induced by an $\omega$-language $L$ (usually given by a DPA \aut{A}) over \Bsq, and a function $f:\Nat\to\Natp$.
(Since we focus on the impact of the function $f$,
we omit $L$ if it is clear from the context and write \Gd{f}.)
The function $f$ imposes a delay (or lookahead) on the moves of \pO.
This means that in round $i$ \pI has to choose $f(i)$
many bits, and \pO chooses one bit, afterwards.
This way the players build up two infinite sequences; \pI builds up
$\alpha = a_0 a_1 \cdots$ and \pO builds up $\beta = b_0 b_1 \cdots$,
respectively. The corresponding play is winning for \pO if and only if
the word $\obinom{a_0}{b_0}\obinom{a_1}{b_1}\obinom{a_2}{b_2}\cdots$
is accepted by \aut{A}.
For a DPA \aut{A}, we say that
\LA is \emph{solvable with finite delay}
if and only if there exists $f:\Nat\to\Natp$ such that \pO wins \GdLA{f}
(analogously for restricted classes of functions).

Observe that the possible strategies for \pO in \Gd{f} correspond precisely
to $h$-delay operators, since \pO must output her $i$th bit after receiving
the next $f(i)$ bits of input. Thus, the question whether there exists
an $h$-delay operator $\lambda$ such that 
$\{\obinom{\alpha}{\lambda(\alpha)} \mid \alpha \in \Bom\} \subseteq L(\aut{A})$
is equivalent to the question whether there exists a winning strategy for
\pO in \Gd{f}.

A basic observation is that winning with delay is a monotone property.
For two functions $f,g: \Nat \to \Natp$ we write
$f\sqsubseteq g$ if and only if $f(i)\leq g(i)\text{ for all } i\in\Nat$.
\begin{rem}\label{rem:winning_monotone}
If \pO wins \Gd{f_0} then she also wins \Gd{f} for each $f\sqsupseteq f_0$.
Analogously, if \pI wins \Gd{g_0} then he also wins \Gd{g},
for each $g\sqsubseteq g_0$.
\end{rem}

\begin{exa}\label{ex:game_with_delay}
Let $L\subseteq\Bsqom$ be given by the $\omega$-regular expression
\[ \obinom{0\:a}{a\:*}\Som+\obinom{1**\:b}{b**\:*}\Som \]
where $a,b\in\B$ and $*$ denotes any bit.
If \pI chooses 0 as his first bit then \pO needs to know $a$,
so she needs delay one in this situation.
Contrary, if \pI chooses $1$ as his first bit then \pO needs delay three
to obtain $b$. Thus, she wins the game with delay three,
but neither with delay two nor one.
\end{exa}

In the next sections we prove our main result (see Theorem~\ref{thm:main}):
\emph{Let \aut{A} be a DPA
with $n$ states, $m$ colors, and let $n' := 2^{(mn)^{2n}}$.
Then, there is a continuous operator $\lambda$
with $\obinom{\al}{\lambda(\al)}\in\LA$ (for all $\al\in\Bom$)
if and only if there is a $(2n'-1)$-delay operator with the same property.}
To obtain this result we show that
\LA is solvable with finite delay
if and only if \LA is solvable with delay $2n'-1$.

%--------------------------------------------Section 4
\section{The Block Game}\label{sec:block_game}
In this section we make the first step in the proof of our main result,
which is to relax the number of bits \pI can choose in each move.
For this reason we introduce a new game \Gp{f}, called the \emph{block game}.

The game \Gp{f} differs from \Gd{f} in two ways.
Firstly, the lengths of the words to be chosen by the players
are decided by \pI, within certain intervals determined by $f$.
Secondly, \pI is one move ahead compared to \Gd{f}.

A play in \Gp{f} is built up as follows: \pI chooses
$u_0\in\B^{[f(0),2f(0)]}$ and $u_1\in\B^{[f(1),2f(1)]}$,
then \pO chooses $v_0\in\B^{|u_0|}$.
In each round thereafter, \ie for $i\geq2$, \pI chooses
$u_i\in\B^{[f(i),2f(i)]}$ and \pO responds
by a word $v_{i-1}\in\B^{|u_{i-1}|}$.
The winning condition is defined as before.

We show that \pI wins the game \Gd{f} for all functions $f$ if and only
if he wins the block game \Gp{f} for all functions $f$. To this end, for
$f:\Nat\to\Natp$, let $f'$ be defined by
$f'(0) := f(0)+f(1)$, and $f'(i) := f(i+1)$ for $i>0$.

\begin{prop}\label{prop:G_f_prime_G_prime_f}
Let $f:\Nat\to\Natp$. If \pI wins \Gd{f'} then he also wins \Gp{f}.
\end{prop}
\begin{proof}
Assume \pI has a winning strategy in \Gd{f'}. For $i\in\Nat$, let
$u_i$ be the words chosen by \pI in \Gd{f'} and $u'_i$ the words
chosen by \pI in \Gp{f}, and analogously $v_i,v'_i$ for \pO.
The winning strategy yields $u_0\in\B^{f'(0)}$ as \pI's first move.
Since $f(0)+f(1)=f'(0)$ we can choose $u'_0u'_1=u_0$ as \pI's first move
in \Gp{f}. \pO answers by $v'_0\in\B^{|u'_0|}$. We can use $v'_0$ to
simulate the moves $v_0,\ldots,v_{|v'_0|-1}$ of \pO in \Gd{f'},
each of which consists of one bit. \pI answers by
$u_1,\ldots,u_{|v'_0|}$ of lengths $f'(1),\ldots,f'(|v'_0|)$.
Since $|v'_0| \geq 1$, the sum $f'(1)+\cdots+f'(|v'_0|)$ is non-empty and
at least $f'(1)=f(2)$. Accordingly, the word $u_1\cdots u_{|v'_0|}$ is
long enough to give $u'_2$ with $f(2) \leq|u'_2| \leq 2f(2)$.
We choose $u'_2$ as the prefix of $u_1\cdots u_{|v'_0|}$ of length $f(2)$.
\pO answers in \Gp{f} by $v'_1$ of length $|u'_1|$, and we can use it to
simulate another $|v'_1|$ rounds in \Gd{f'}. Thereby, we obtain enough
bits to give $u'_3$, and so on. This way, we build up the same plays
in \Gd{f'} and \Gp{f}. Since \pI wins \Gd{f'}, he also wins \Gp{f}.
\end{proof}

For $f:\Nat\to\Natp$, let $f''$ be inductively defined by $f''(0):=f(0)$
and \[ f''(i+1) \, := \, \sum\limits_{j=0}^{2(f''(0)+\ldots+f''(i))}f(j). \]
\begin{prop}\label{prop:G_prime_f_prime_prime_G_f}
Let $f:\Nat\to\Natp$. If \pI wins \Gp{f''} then he also wins \Gd{f}.
\end{prop}
\begin{proof}
Assume \pI has a winning strategy in \Gp{f''}. For $i\in\Nat$, let
$u'_i$ be the words chosen by \pI in \Gp{f''} and $u_i$ the words
chosen by \pI in \Gd{f}, and analogously $v'_i,v_i$  for \pO.
\pI's winning strategy yields $u'_0\in\B^{[f''(0),2f''(0)]}$ and
$u'_1\in\B^{[f''(1),2f''(1)]}$ as his first move in \Gp{f''}.
For $i\in\Nat$, let $d'_i$ be the length of $u'_i$.
Since
\[ d'_0+d'_1\geq f''(0)+f''(1)=f(0)+\sum\limits_{j=0}^{2f''(0)}f(j), \]
we can give the moves $u_0,\ldots,u_{d'_0}$ of \pI in \Gd{f}. This
yields \pO's answers $v_0,\ldots,v_{d'_0-1}$, \ie $d'_0$ bits. We
can use them to simulate $v'_0$, \ie \pO's first move in
\Gp{f''}. \pI's winning strategy yields $u'_2$ of length $f''(2)\leq
d'_2\leq2f''(2)$. We need to give another $d'_1$ moves of \pI in
\Gd{f} to obtain \pO's answers $v_{d'_0},\ldots,v_{d'_0+d'_1-1}$. For
that we need $f(d'_0+1)+\ldots+f(d'_0+d'_1)$ bits. With $u'_2$ in our
hands we can give these moves, because
\begin{center}
\begin{tabular}{rcccl}
$d'_2$ & $\geq$ & $f''(2)$ & $=$ & $f(0)+\ldots+f(2f''(0)+2f''(1))$\\
	& & & $\geq$ & $f(0)+\ldots+f(d'_0+d'_1)$\\
	& & & $\geq$ & $f(d'_0+1)+\ldots+f(d'_0+d'_1)$.
\end{tabular}
\end{center}
Iterating this we obtain the same plays built up in \Gp{f''} and
\Gd{f}. Since \pI wins \Gp{f''}, he also wins \Gd{f}.
\end{proof}

The following corollary of Propositions \ref{prop:G_f_prime_G_prime_f} and
\ref{prop:G_prime_f_prime_prime_G_f}, which follows by taking functions of
the form $f'$ in the one direction and of the form $f''$ in the other,
is the first step in our proof.

\begin{cor}\label{cor:equivalence_all_Gf_all_G_prime_f}
Let \aut{A} be a DPA. Then the following are equivalent:
\begin{enumerate}[\em(1)]
\item For all $f:\Nat\to\Natp$ \pI wins \GdLA{f}.
\item For all $f:\Nat\to\Natp$ \pI wins \GpLA{f}.
\end{enumerate}
\end{cor}

%--------------------------------------------Section 5
\section{The Semigroup Game}\label{sec:semigroup_game}

In this section we introduce a game which is independent of particular delays.
To define it, we extract from a DPA \aut{A} two equivalence relations,
one for each player, such that the moves of the players are equivalence classes of
these relations. The first one (for \pO) is denoted $\sim$ and induces
a finite semigroup on~\Bsqst. The second one (for \pI)
is denoted $\approx$ and ranges over \Bst.
Roughly speaking, two (pairs of) words are equivalent
if they effect the same behavior on~\aut{A}.

Our approach to transform parity automata into finite semigroups
is similar to the constructions presented in \cite{PP95SemInfinite,Pin95FiniteSem}.
Let $\aut{A}=(Q,q_0,\delta,c)$ be a DPA over \Bsq.
We use the semiring $\mathcal{S}:=(\{\bot\}\cup c(Q),+,\cdot)$ in
which addition is defined as maximum, \ie $x+y:=\max(x,y)$ with $\bot$
being the least element, and multiplication is defined as follows:
\[ x\cdot y:=\begin{cases}
    \max(x,y) & \text{if }x\neq\bot\text{ and }y\neq\bot\\
    \bot & \text{otherwise}
  \end{cases}
\]
Note that the set $\Leq:=\Bsqst$, \ie the set of pairs of words of equal length, is a regular language.
With each pair $\obinom{u}{v} \in \Leq$ we associate a matrix $\mu\obinom{u}{v}$
of size $|Q|^2$ with entries in $\mathcal{S}$, \ie $\mu\obinom{u}{v}\in\mathcal{S}^{Q\times Q}$,
defined as follows:
\[ \mu\obinom{u}{v}_{p,q} := \begin{cases}
      \bot & \text{if }\delta^*\Big(p,\binom{u}{v}\Big)\neq q\\
      \max \{c(\pi)\} & \text{if }\delta^*\Big(p,\binom{u}{v}\Big)=q
        \text{ and }\pi\text{ is the associated \aut{A}-path}
   \end{cases}
\]
Observe that $\mathcal{S}^{Q\times Q}$ induces a finite semigroup and
$\mu\obinom{u}{v} \cdot \mu\obinom{u'}{v'} = \mu\obinom{uu'}{vv'}$.
Let $\sim$ be the equivalence relation on \Leq defined by:
$\obinom{u}{v}\sim\obinom{u'}{v'}$ if and only if $\mu\obinom{u}{v}=\mu\obinom{u'}{v'}$.
For each $\obinom{u}{v}$, the equivalence class $\big[\obinom{u}{v}\big]$ is identified by
a matrix $\mu \in \mathcal{S}^{Q\times Q}$. Since $\mathcal{S}$ and $Q$
are finite, $\mathcal{S}^{Q\times Q}$ is finite as well, and so the
relation $\sim$ has finite index, \ie it has finitely many equivalence classes.
We denote the index of $\sim$ by $\Index(\sim)$. Note that $\Leq/_{\sim}$
induces a finite semigroup, and $\mu$ is a semigroup morphism from
$(\Leq/_{\sim},\cdot)$ to $(\mathcal{S}^{Q\times Q},\cdot)$.

\begin{lem}\label{lem:sim_regular_class}
Let $\obinom{u}{v}\in\Leq$. Then, the set $\big[\obinom{u}{v}\big]$
is a regular $*$-language over \Bsq.
\end{lem}

\begin{proof}
We construct an automaton recognizing $\big[\obinom{u}{v}\big]$ as follows:
First, we construct for all $p,q\in Q,k\in c(Q)$ the automaton
$\aut{A}_{p,q,k}$ recognizing the set of all words that induce a path
from $p$ to $q$ in \aut{A} where $k$ is the highest color seen on that
path. The idea for this construction is to simulate the behavior of
\aut{A} while memorizing the highest color seen. To this end, define
$\aut{A}_{p,q,k}:=(c(Q)\times Q,\Bsq,(c(p),p),\delta',\{(k,q)\})$
where
\[ \delta'\Big((k',p'),\ \binom{x}{y}\Big) :=
  \Big(\max \Big\{ k',c\Big( \delta\Big(p', \binom{x}{y}\Big)\Big) \Big\},\
      \delta \Big(p', \binom{x}{y}\Big)\Big)
\]
for all $k'\in c(Q),p'\in Q,x,y\in\B$. The automaton starts in the state
$(c(p),p)$ and simulates the behavior of \aut{A} on its input. If it
stops in state $(k,q)$ then it accepts. The automaton $\aut{A}_{[\obinom{u}{v}]}$
is then obtained as the intersection of all $\aut{A}_{p,q,k}$ for $p,q,k$
such that  $\mu\obinom{u}{v}_{p,q}=k$. 
\end{proof}

Since $\sim$ has finite index, we can find automata for all equivalence
classes of $\sim$ in the following way: For $r\in\Nat$, let
$\aut{A}_1,\ldots,\aut{A}_r$ be the automata already constructed.
Then $\sim$ has index $r$ if and only if
$\bigcup_{i=1,\ldots,r}L(\aut{A}_i)=\Leq$.
This equality can be effectively checked, and if this test fails,
then we repeat the construction with a word contained in
$\Leq\setminus\bigcup_{i=1,\ldots,r}L(\aut{A}_i)$.

Let $\approx$ be
the equivalence relation on \Bst defined by
\[
 u \approx u' :\iff \forall \Big[\binom{u_0}{v_0}\Big]:
   \Big( \exists v: \binom{u}{v} \in \Big[ \binom{u_0}{v_0} \Big]
     \iff \exists v': \binom{u'}{v'} \in \Big[ \binom{u_0}{v_0} \Big]\Big).
\]
For $u\in\Bst$, the $\approx$-equivalence class of $u$, denoted $[u]$,
can be identified with a subset of the set of all $\sim$-classes.
Since $\sim$ has finite index, we get that $\approx$ has finite index as well;
more precisely it holds $\Index(\approx)\leq2^{\Index(\sim)}$.

\begin{lem}\label{lem:approx_regular_class}
Let $u\in\Bst$. Then, the set $[u]$ is a regular $*$-language over \B.
\end{lem}

\proof
We construct an automaton recognizing the language $[u]$ as follows:
First, we have to check for which $\sim$-classes $\big[\obinom{u_0}{v_0}\big]$
there exists $v\in\B^{|u|}$ such that
$\obinom{u}{v} \in \big[ \obinom{u_0}{v_0} \big]$. Let \aut{B} be a DFA
recognizing $\big[\obinom{u_0}{v_0} \big]$. We take the projection on the
first component (deleting the second component from the transitions of \aut{B})
and test whether the resulting automaton, say $\aut{B}'$, accepts $u$.
If we do the same for all $\sim$-classes,
then we obtain $r$ automata $\aut{B}'_1,\ldots,\aut{B}'_r$ accepting
$u$, and $s$ automata $\aut{B}'_{r+1},\ldots,\aut{B}'_{r+s}$ not accepting $u$,
where $r+s = \Index(\sim)$. From these automata we can effectively construct
an automaton for $[u]$, because \label{pageref:definition_L_A_u}
\[ [u] = \bigcap\limits_{i=1,\ldots,r}L(\aut{B}'_i) \cap
         \bigcap\limits_{j=r+1,\ldots,r+s}\overline{L(\aut{B}'_j)}.
\rlap{\hbox to 215 pt{\hfill\qEd}}
\]
We now define the game \GSG (induced by a DPA \aut{A} over \Bsq) where the moves of
the players are classes from $\Bst/_{\approx}$ and $\Leq/_{\sim}$, respectively.
Accordingly, we call \GSG the \emph{semigroup game} of \aut{A}.

The game \GSG is defined similar to the block game \Gpempty.
The difference is that the players do not choose concrete words but the
respective classes from the relations $\sim$ and $\approx$. A play is
built up as follows: \pI chooses infinite classes
$[u_0],[u_1]\in\Bst/_{\approx}$, then \pO chooses a class
$\big[\obinom{u_0}{v_0}\big] \in \Leq/_{\sim}$. In each round thereafter,
\ie for $i\geq2$, \pI chooses an infinite class
$[u_i]\in\Bst/_{\approx}$ and \pO chooses a class
$\big[\obinom{u_{i-1}}{v_{i-1}}\big]\in\Leq/_{\sim}$.
A play is winning for \pO if and only if
$\obinom{u_0}{v_0}\obinom{u_1}{v_1}\obinom{u_2}{v_2}\cdots$ is accepted
by \aut{A}.

Note that $\Bst/_{\approx}$ contains at least one infinite class and that
for each class $[u]$ there exists at least one class in $\Leq/_{\sim}$
associated with $[u]$ (by the definition of~$\approx$).
Hence, both players can always move.
Furthermore, the winning condition of \GSG is well-defined
because acceptance of \aut{A} is independent of representatives:
If $\big[\obinom{u_i}{v_i}\big] = \big[\obinom{u'_i}{v'_i}\big]$ for all $i\in\Nat$, then
$\obinom{u_0}{v_0}\obinom{u_1}{v_1}\cdots\in\LA\iff\obinom{u'_0}{v'_0}\obinom{u'_1}{v'_1}\cdots\in\LA$.

\GSG\label{pageref:size_of_GSG} can be modeled by a parity game on a graph of size $O(2^{2(mn)^n}mn)$.
(Thus, its winner is computable~\cite{GTW02AutLogInfGam}.)
In the vertices we keep track of the $\approx$-classes recently chosen by \pI,
a color depending on the course of the play and
the current state $q$ of \aut{A}. The vertex reached by
a move $\big[\obinom{u}{v}\big]$ of \pO is colored by $\mu\obinom{u}{v}_{q,q'}$,
where $q'$ is the state reached in \aut{A} from $q$ when reading $\obinom{u}{v}$.

%--------------------------------------------Section 6
\section{Connecting the Block Game and the Semigroup Game}\label{sec:connection_block_semigroup_game}
In this section we show that \pI wins the block game \Gp{f} for all
functions $f:\Nat\to\Natp$ if and only if he wins the semigroup game \GSG.
This completes the reduction and also yields the proof of our main result.

The basic idea of the proof of Theorem~\ref{thm:equivalence_all_f_greater_than_f_0_GSG} (see below)
is, for arbitrary $f$, to simulate the moves of the players in \Gp{f} by the corresponding
equivalence classes of the relations $\sim$ and $\approx$, respectively, and vice versa.
For the last-mentioned direction,
one has the problem whether a class $[u_i]$ contains an
appropriate representative, \ie one of length between $f(i)$ and $2f(i)$.
We use Lemma~\ref{lem:length} to show that there exists a particular $f$
such that each function $g$ with $g\sqsupseteq f$ indeed has this property.
Then, the following lemma completes the proof.

\begin{lem}\label{lem:equivalence_all_f_all_f_greater_than_f_0}
\pI wins \Gp{f} for all functions $f:\Nat\to\Natp$ if and only if there exists
a function $g_0 : \Nat \to \Natp$ such that \pI wins \Gp{g} for
all $g\sqsupseteq g_0$.
\end{lem}
\begin{proof}
The direction from left to right is immediate. For the converse,
recall first that the block game \Gp{f} is determined for each $f$.
Assume there exists $f_0$ such that \pI does not win \Gp{f_0}.
Determinacy yields that \pO wins \Gp{f_0}.
By Proposition~\ref{prop:G_f_prime_G_prime_f} \pO wins \Gd{f'_0}, and
from Remark~\ref{rem:winning_monotone} it follows that she also wins \Gd{f} for
all $f\sqsupseteq f'_0$. Proposition~\ref{prop:G_prime_f_prime_prime_G_f}
yields that \pO wins \Gp{f''}, for all $f\sqsupseteq f'_0$. Towards
a contradiction, let $g_0$ be a function such that \pI wins \Gp{g} for
all $g\sqsupseteq g_0$, and let $f_*$ be the maximum of $g_0$ and
$f'_0$, \ie for all $i\in\Nat$
  \[ f_*(i):=\max\{g_0(i),f'_0(i)\}. \]
Since $f_*\sqsupseteq f'_0$ it holds that \pO wins
\Gp{f''_*}. However, since $f''_*\sqsupseteq f_*\sqsupseteq g_0$ \pI
must win \Gp{f''_*}, by assumption. This yields a contradiction which
means that $g_0$ cannot exist.
\end{proof}

Lemma~\ref{lem:equivalence_all_f_all_f_greater_than_f_0} and the next
theorem establish the second step of our reduction.

\begin{thm}\label{thm:equivalence_all_f_greater_than_f_0_GSG}
\pI wins \GSG if and only if there is a function $f:\Nat\to\Natp$ such that
\pI wins \Gp{g} for all $g\sqsupseteq f$.
\end{thm}

\begin{proof}
We start with the direction from right to left. Let $f:\Nat\to\Natp$
be a function such that \pI wins \Gp{g} for all $g\sqsupseteq f$. We
define a function $g_0$ such that $g_0\sqsupseteq f$ and each word of
length $g_0(i)$ is contained in an infinite $\approx$-class, for all $i\in\Nat$.
To this end, let $d'$ be the length of a longest word in all finite
$\approx$-classes\footnote{If $\approx$ has no finite
equivalence class, then we define $d':=0$.} and define, for all $i\in\Nat$,
$g_0(i):=\max\{f(i),d'+1\}$.

Since $g_0\sqsupseteq f$, \pI wins \Gp{g_0} by assumption, and a winning strategy yields his first
two moves $u_0,u_1$. Both $[u_0]$ and $[u_1]$ are infinite, and so he can
choose them in \GSG. We simulate \pO's answer $\big[\obinom{u_0}{v_0}\big]$
by choosing $v_0$ in \Gp{g_0}, and \pI's winning strategy yields $u_2$ with $[u_2]$ being
infinite. Choosing $[u_2]$ in \GSG we obtain \pO's next move
$\big[\obinom{u_1}{v_1}\big]$, and so on.

We argue that the plays built up have the same maximal color occurring
infinitely often. It suffices to show that in both plays a move of \pO
leads \aut{A} to the same state, via paths with equal maximal
color. Then, the rest follows by induction. Let $q_i$ be the current
state of \aut{A} and $u_i,u_{i+1}$ be the words chosen by \pI. If \pO
chooses $\big[\obinom{u_i}{v_i}\big]$ in \GSG, then we reach the state
$q_{i+1}:=\delta^*\big(q_i,\obinom{u_i}{v_i}\big)$
via the maximal color
$\mu\obinom{u_i}{v_i}_{q_i,q_{i+1}}$. The state $q_{i+1}$ is well-defined
because from $q_i$ every $\obinom{u'_i}{v'_i} \in \big[\obinom{u_i}{v_i}\big]$
leads \aut{A} to the same state, though via different paths, but with the
same maximal color. In \Gp{g_0} \pO chooses $v_i$. As in \GSG, we reach
the state $q_{i+1}$ via the maximal color
$\mu\obinom{u_i}{v_i}_{q_i,q_{i+1}}$.

Conversely, assume that \pI wins \GSG. Let
$\aut{A}_1,\ldots,\aut{A}_r$ be automata recognizing all the
$\approx$-classes, and $n'$ the maximal number of states among these
automata, \ie $n':=\max\{n_1,\ldots,n_r\}$, where $n_j$ is the number of states of $\aut{A}_j$ ($j=1,\ldots,r$).
Let $f$ be the constant function with $f(i):=n'$ for all $i\in\Nat$.
We first show that \pI wins \Gp{f}:
\pI's winning strategy in \GSG yields $[u_0],[u_1]$.
Since $[u_0],[u_1]$ are infinite, we can apply Lemma~\ref{lem:length}.
Accordingly, each $\aut{A}_j$ accepts a word of length between $f$ and
$f+n_j$ and thus between $f$ and $2f$, because $n_j\leq f$.\footnote{To simplify matters we write $f$ instead of $f(i)$.}
Hence, we can assume \wlofg that $f\leq|u_0|,|u_1|\leq2f$.
\pI chooses $u_0,u_1$ in \Gp{f} and \pO answers by a word $v_0$ with
$|v_0|=|u_0|$. We simulate this move by $\big[\obinom{u_0}{v_0}\big]$ in \GSG
and obtain \pI's answer $[u_2]$, so the next move of \pI in \Gp{f} is
$u_2$ (for appropriate $u_2$). \pO chooses $v_1$ with $|v_1|=|u_1|$,
and so on.

The plays built up this way have the same maximal color occurring
infinitely often, using the same inductive argument as above. Starting
at $q_i$, \pO's move $v_i$ in \Gp{f} has the same effect as the
corresponding move $\big[\obinom{u_i}{v_i}\big]$ in \GSG, \ie we reach
the state $q_{i+1}:=\delta^*\big(q_i,\obinom{u_i}{v_i}\big)$ via
the maximal color $\mu\obinom{u_i}{v_i}_{q_i,q_{i+1}}$.

We complete the proof by showing that \pI wins \Gp{g} for all $g\sqsupseteq f$.
Let $|[a,b]|:=b-a$ be the size of the interval
$[a,b]$. If $g\sqsupseteq f$, then (since $|[f,2f]|=n'$) it holds
$|[g(i),2g(i)]|\geq n'$, for all $i\in\Nat$. Hence, to win \Gp{g} \pI
simply chooses longer representatives of the $\approx$-classes than in
\Gp{f}.
\end{proof}

A thorough analysis of the constructions of the $\sim$-classes and
$\approx$-classes, respectively, yields an upper bound for $n'$. Let
$n$ be the number of states of \aut{A} and $m$ the number of
colors. Let $u,v\in\Bst$ with $|u|=|v|$. Since \aut{A} is deterministic,
there is exactly one entry distinct from $\bot$ in each of the $n$
rows of $\mu\obinom{u}{v}$, and $\aut{A}_{p,q,k}$ has at most $mn$
states. Hence, each $\aut{A}_{[\obinom{u}{v}]}$ has at most $(mn)^n$
states, \ie as many as the product of $n$ (deterministic) automata
of size $mn$. To obtain an automaton for a class $[u]$ we have to intersect
$\Index(\sim)$ languages (\cf page~\pageref{pageref:definition_L_A_u}).
By the same argument as above, there are at most $(mn)^n$ possible
matrices identifying all the $\sim$-classes. Since our construction
includes determinization, we obtain each $\aut{A}_{[u]}$ having
at most $k$ states, where
\[ k\leq(2^{(mn)^{n}})^{(mn)^n}=2^{(mn)^{2n}}. \]
Next, we obtain our main result showing that in regular games constant delay
is sufficient for \pO to win, if she can win with delay at all.
Recall that we write $\const{d}$ for the constant delay function,
$\const{d}(0)=d+1$ and $\const{d}(i)=1$ for $i>0$.

\begin{lem}\label{lem:bound}
Let $n'$ be as in the proof of Theorem~\ref{thm:equivalence_all_f_greater_than_f_0_GSG}.
Then, \pO wins \GSG if and only if \pO wins \Gd{\const{2n'-1}}.
\end{lem}

\begin{proof}
Define $f(i):=n'$ for all $i\in\Nat$ and let $w$ of length $d'$ be a longest word in all
finite $\approx$-classes. Moreover, let $L(\aut{A}')=[w]$, where $\aut{A}'$ has $n$ states.
Then we have $d'<n$. Otherwise, the run of $\aut{A}'$ on $w$ had a loop,
which is a contradiction to the finiteness of $L(\aut{A}')$.
Since $n\leq n'$ we get $d'<n'$ and so $d'+1\leq n'$.
Thus, each $\approx$-class containing a word of length at least $f$ is infinite.

Assume that \pO wins \GSG. We first show that \pO wins \Gp{f}.
Let $u_0,u_1$ with $n'\leq|u_0|,|u_1|\leq2n'$ be the first move of \pI in
\Gp{f}. By the above remarks $[u_0],[u_1]$ are infinite, and we can
simulate $[u_0],[u_1]$ in \GSG. \pO's winning strategy in \GSG yields
$\big[\obinom{u_0}{v_0}\big]$ for some suitable $v_0$. Let him choose
$v_0$ in \Gp{f}. Then \pI chooses $u_2$ and we simulate $[u_2]$ in
\GSG, and so on.

As in the proof of Theorem~\ref{thm:equivalence_all_f_greater_than_f_0_GSG},
we obtain plays with the same maximal color occurring infinitely often, and so
\pO wins \Gp{f}. Simulating a winning strategy for \Gp{f} she also wins
\Gd{\const{2n'-1}}. The factor $2$ comes from the fact that we need at least
$2n'$ bits when simulating \pI's first move in \Gp{f}.

Conversely, let \pO win \Gd{\const{2n'-1}} and $g(i):=2n'$, for all
$i\in\Nat$. Since $g\sqsupseteq\const{2n'-1}$, \pO wins \Gd{g}.
Then, by Proposition~\ref{prop:G_prime_f_prime_prime_G_f},
she also wins \Gp{g''}. Given a winning strategy for \pO in \Gp{g''}
we can specify one for her in \GSG as follows: A move $[u_i]$ of \pI is
simulated by $u_i$ in \Gp{g''}, for $g''(i)\leq|u_i|\leq2g''(i)$.
(By Lemma~\ref{lem:length}, an appropriate representative $u_i$ must exist
because $g''\sqsupseteq g$, and so $|[g''(i),2g''(i)]|\geq n'$ for all
$i\in\Nat$.) We use \pO's answer $v_{i-1}$ to choose
$\big[\obinom{u_{i-1}}{v_{i-1}}\big]$ in \GSG.
This yields a play winning for \pO in \GSG.
\end{proof}

With Corollary~\ref{cor:equivalence_all_Gf_all_G_prime_f}, Lemma~\ref{lem:equivalence_all_f_all_f_greater_than_f_0}
and Theorem~\ref{thm:equivalence_all_f_greater_than_f_0_GSG} we have shown that the problem whether
\LA is solvable with finite delay is reducible to the question whether
\pO wins \GSG. Finally, Lemma~\ref{lem:bound} shows that
\LA is solvable with finite delay if and only if it is solvable with constant delay.

\begin{thm}\label{thm:main}
Let \aut{A} be a DPA over \Bsq. Then, \LA is solvable with finite delay
if and only if \LA is solvable with delay $2n'-1$.
There is a continuous operator $\lambda$ such that
$\{\obinom{\al}{\lambda(\al)}\mid\al \in \Bom\} \subseteq L(\aut{A})$
if and only if there is a $(2n'-1)$-delay operator with the same property.
\end{thm}

Assuming that \aut{A} has $n$ states and $m$ colors
we can bound the number of vertices of \GSG by $2^{2(mn)^n+1}mn$.
Since it requires only $m$ colors,
its winner can be computed in time $O((2^{2(mn)^n+1}mn)^m)$ \cite{Sch07ParGameBigSteps}.

\begin{cor}\label{cor:decide_bounded_delay}
Let \aut{A} be a DPA over \Bsq. The problem
whether \LA is solvable with finite delay
and the problem whether there is a continuous operator $\lambda$ with
$\{\obinom{\al}{\lambda(\alpha)}\mid \alpha \in \Bom\} \subseteq \LA$
are in $2$\textsc{ExpTime}.
\end{cor}

%--------------------------------------------Section 7
\section{Lookahead in Non-Regular Games}\label{sec:pusdown_games_delay}

In this section we show that the above results do not hold for context-free
$\omega$-languages (\CFLom, for an introduction see \eg \cite{CG78OmegaComputDetPushMach}).
Let us first recall that it is undecidable whether a context-free
$\omega$-language $L\subseteq\Bom$ is universal, \ie whether $L=\Bom$ holds.

\begin{thm}[see also \cite{Fin01TopProp}]\label{thm:finkel_context_free_gale_stewart_undecidable}
Let $L\subseteq\Bsqom$ be a context-free $\omega$-language.
Then, it is undecidable whether there exists $f$ such that \pO wins \GdL{f}.
\end{thm}
\begin{proof}
We make a reduction from the universality problem for context-free $\omega$-languages.
Let $L_I\in\CFLom$ and $L:=\big\{\obinom{\al}{\be} \mid \al\in L_I, \beta\in \Bom \big\}$.
If $L_I$ is universal then $L$ is universal as well, and \pO wins
with any $f$.
Conversely, if $L_I$ is not universal, then \pI wins by playing a word $\al\notin L_I$.
There is no response \be such that $\obinom{\al}{\be}\in L$, therefore
\pO looses with each $f$.
Altogether, $L_I$ is universal if and only if
there exists $f$ such that \pO wins~\GdL{f}.
\end{proof}
It has recently been shown \cite{FLZ11} that the same holds for
\emph{deterministic} $\omega$-context-free specifications, but in that
case at least establishing the winner of the standard game \Gd{\const{0}}
is decidable \cite{Wal96PushProc}.

In addition to undecidability for the general case, we show that there exist
context-free specifications which are solvable with finite delay,
but not with constant delay.

\begin{exa}\label{ex:non_regular_not_constant_delay}
Let $L\subseteq\Bsqom$ be defined such that if \pI chooses an $\omega$-word
of the form $\al=1^{2m_0}0^{n_0}1^{2m_1}0^{n_1}\cdots$, for $m_i,n_i\in\Natp$, then
\pO wins if and only if he answers by $\beta=1^{m_0}0^{m_0+n_0}1^{m_1}0^{m_1+n_1}\cdots$.
This means \pO's $i$th block of $1$s must have exactly half the length of
\pI's $i$th block of $1$s, and both blocks must start at the same position.
If \al is not of the above form, then \pO wins as well.

The language $L$ is recognized by a deterministic $\omega$-pushdown automaton.
As long as the input is $\obinom{1}{1}$, we push a symbol on the stack.
If we read the first $\obinom{1}{0}$ after $\obinom{1}{1}$,
we start to pop symbols from the stack. If we reach the initial stack
symbol at the same time as we read the first $\obinom{0}{0}$
after $\obinom{1}{0}$ then we are satisfied and visit a final state.

Observe that \pO wins \GdL{f}, if $f(i):=2$ for all $i\in\Nat$.
When she has to give her $i$th bit $\be_i$ she already knows \pI's $(2i)$th bit $\al_{2i}$,
and that is enough to decide whether to play $0$ or $1$.

Let us show that $L$ is not solvable with constant delay.
Towards a contradiction, assume \pO wins \Gd{\const{d}} for some $d\in\Nat$.
We construct a winning strategy for \pI in \Gd{\const{d}} as follows:
\pI chooses $1^{d+1}$ as initial move and 1 as each of his $d$ subsequent moves.
\pO must answer each of these $d+1$ moves by choosing $1$. Otherwise, she loses immediately.
Afterwards, \pI chooses another $1$ to complete his block of $1$s to even length.
(After this move, \pI has chosen exactly twice as many $1$s as \pO.)
Whatever \pO answers, say $b$, \pI wins by choosing $1-b$ next.
This is due to the fact that the block of $1$s chosen by \pO gets either
too short or too long. 
\end{exa}

%--------------------------------------------Section 8
\section{Conclusion}\label{sec:conclusion}

In this paper we introduced and compared strategies
with different kinds of lookahead in regular infinite games.
We showed that continuous strategies can be reduced to
uniformly continuous strategies of a special form,
namely strategies with constant lookahead.
This result is a first step into a wider -- and it
seems rather unexplored -- topic. Let us mention some
aspects. First, it is straightforward to present
the results in a set-up that is symmetric in the two players.
We also skipped here a lower bound proof for the
double exponential size in Theorem \ref{thm:main}.
It is also possible to think of ``infinite lookahead'' where,
for instance, the second player may use information about
the first player's sequence up to a partition of the
space of sequences into regular sets. Moreover, while
basic questions about lookahead in context-free games have
recently been answered, some problems for visibly pushdown
winning conditions remain open, cf. \cite{FLZ11}.


\begin{thebibliography}{10}

\bibitem{BSW03PushGamUnboundRegCond}
Alexis-Julien Bouquet, Olivier Serre, and Igor Walukiewicz.
\newblock Pushdown games with unboundedness and regular conditions.
\newblock volume 2914 of {\em LNCS}, pages 88--99. Springer, 2003.

\bibitem{BL69SolSeqCondFinStateStr}
J.~Richard B{\"u}chi and Lawrence~H. Landweber.
\newblock Solving sequential conditions by finite-state strategies.
\newblock {\em Transactions of the AMS}, 138:295--311, 1969.

\bibitem{Cach03HighOrdPushAutCaucalHier}
Thierry Cachat.
\newblock Higher order pushdown automata, the caucal hierarchy of graphs and
  parity games.
\newblock In Jos C.~M. Baeten, Jan~Karel Lenstra, Joachim Parrow, and
  Gerhard~J. Woeginger, editors, {\em ICALP}, volume 2719 of {\em LNCS}, pages
  556--569. Springer, 2003.

\bibitem{CG78OmegaComputDetPushMach}
Rina~S. Cohen and Arie~Y. Gold.
\newblock Omega-computations on deterministic pushdown machines.
\newblock {\em Journal of Computer and System Sciences}, 16(3):275--300, 1978.

\bibitem{EM69SeqBoolEq}
Shimon Even and Albert~R. Meyer.
\newblock Sequential boolean equations.
\newblock {\em IEEE Transactions on Computers}, C-18(3):230--240, 1969.

\bibitem{Fin01TopProp}
Olivier Finkel.
\newblock Topological properties of omega context-free languages.
\newblock {\em Theoretical~Computer~Science}, 262(1-2):669--697, 2001.

\bibitem{FLZ11}
Wladimir Fridman, Christof L{\"o}ding, and Martin Zimmermann.
\newblock Degrees of lookahead in context-free infinite games.
\newblock In {\em Proceedings of CSL~'11}, volume~12 of {\em LIPIcs}, pages
  264--276. Schloss Dagstuhl -- Leibniz-Zentrum f{\"u}r Informatik, 2011.

\bibitem{GTW02AutLogInfGam}
Erich Gr{\"a}del, Wolfgang Thomas, and Thomas Wilke, editors.
\newblock {\em Automata, Logics and Infinite Games}, volume 2500 of {\em LNCS}.
\newblock Springer, 2002.

\bibitem{HL72FinDelSol}
Frederick~A. Hosch and Lawrence~H. Landweber.
\newblock Finite delay solutions for sequential conditions.
\newblock In M.~Nivat, editor, {\em Automata, Languages and Programming}, pages
  45--60, Paris, France, 1972. North-Holland, Amsterdam.

\bibitem{Mosch80DST}
Yiannis~N. Moschovakis.
\newblock {\em Descriptive Set Theory}, volume 100 of {\em Studies in Logic and
  the Foundations of Mathematics}.
\newblock North-Holland Publishing Company, 1980.

\bibitem{PP95SemInfinite}
Dominique Perrin and Jean{-}{\'E}ric Pin.
\newblock Semigroups and automata on infinite words.
\newblock In J.~Fountain, editor, {\em NATO Advanced Study Institute {\it
  Semigroups, Formal Language and Groups}}, pages 49--72. Kluwer academic
  publishers, 1995.

\bibitem{Pin95FiniteSem}
Jean{-}\'{E}ric Pin.
\newblock Finite semigroups and recognizable languages: An introduction, 1995.

\bibitem{Sch07ParGameBigSteps}
Sven Schewe.
\newblock Solving parity games in big steps.
\newblock In Vikraman Arvind and Sanjiva Prasad, editors, {\em FSTTCS}, volume
  4855 of {\em LNCS}, pages 449--460. Springer, 2007.

\bibitem{TL93LogSpecInfComp}
Wolfgang Thomas and Helmut Lescow.
\newblock Logical specifications of infinite computations.
\newblock In J.~W. de~Bakker, W.~P. de~Roever, and G.~Rozenberg, editors, {\em
  REX School/Symposium}, volume 803 of {\em LNCS}, pages 583--621. Springer,
  1993.

\bibitem{TB73FinAutBehSynth}
Boris~A. Trakhtenbrot and Janis~M. Barzdin.
\newblock {\em Finite Automata, Behavior and Synthesis}.
\newblock North Holland, Amsterdam, 1973.

\bibitem{Wal96PushProc}
Igor Walukiewicz.
\newblock Pushdown processes: Games and model checking.
\newblock volume 1102 of {\em LNCS}, pages 62--74. Springer, 1996.

\end{thebibliography}
\end{document}